\documentclass{article}
\usepackage{abstract}
\providecommand{\keywords}[1]{\textbf{Keywords:} #1}

\usepackage{arxiv}

\usepackage[utf8]{inputenc}
\usepackage[T1]{fontenc}
\usepackage{hyperref}
\usepackage{url}
\usepackage{booktabs}
\usepackage{amsfonts}
\usepackage{amsmath}
\usepackage{amssymb}
\usepackage{amsthm}
\usepackage{nicefrac}
\usepackage{microtype}
\usepackage{graphicx}
\usepackage{natbib}
\usepackage{doi}
\usepackage{tikz}
\usetikzlibrary{arrows.meta, positioning, calc, shapes.geometric, fit, backgrounds}
\usepackage{xcolor}
\usepackage{array}
\usepackage{caption}
\usepackage{subcaption}
\usepackage{algorithm}
\usepackage{algpseudocode}
\usepackage{enumitem}
\usepackage{tabularx}
\usepackage{multirow}

\tikzset{
  gridcell/.style={
    draw=black, thick,
    minimum width=1.0cm, minimum height=1.0cm,
    rounded corners=4pt,
    font=\ttfamily\bfseries\large,
    align=center
  },
  emptyCell/.style={gridcell, fill=gray!8},
  instrCell/.style={gridcell, fill=blue!15},
  ctrlCell/.style={gridcell, fill=orange!25},
  haltCell/.style={gridcell, fill=red!20},
  waitCell/.style={gridcell, fill=gray!15},
  ipArrow/.style={-{Stealth[scale=1.4]}, red, very thick},
  loopArrow/.style={-{Stealth[scale=1.2]}, purple!70!black, thick, dashed},
  labelStyle/.style={font=\small\sffamily, text=gray!70!black},
  coordLabel/.style={font=\tiny\ttfamily, text=gray!60}
}

\theoremstyle{definition}
\newtheorem{definition}{Definition}[section]
\newtheorem{example}{Example}[section]
\theoremstyle{plain}
\newtheorem{theorem}{Theorem}[section]
\newtheorem{proposition}{Proposition}[section]

\newtheorem{corollary}{Corollary}[section]
\theoremstyle{remark}
\newtheorem{remark}{Remark}[section]

\newcommand{\GP}{\mathcal{P}}           
\newcommand{\Dom}{\mathcal{D}}          
\newcommand{\Inst}{\mathcal{I}}         
\newcommand{\Val}{\mathcal{V}}          
\newcommand{\ZZ}{\mathbb{Z}}            
\newcommand{\RR}{\mathbb{R}}            
\newcommand{\dir}{\mathit{dir}}
\newcommand{\ip}{\mathit{ip}}
\newcommand{\push}{\mathsf{push}}
\newcommand{\pop}{\mathsf{pop}}
\newcommand{\IP}{\mathit{IP}}
\newcommand{\AS}{\mathit{AS}}
\newcommand{\DS}{\mathit{DS}}
\newcommand{\CDLL}{\mathit{CDLL}}
\newcommand{\Prim}{\mathit{prim}}
\newcommand{\Sec}{\mathit{sec}}
\newcommand{\Ter}{\mathit{ter}}
\newcommand{\instr}[1]{\texttt{#1}}

\title{Grid Programs: A Two-Dimensional, Variable-Free Model of Computation}

\author{ \href{https://orcid.org/0000-0001-8231-5687}{\includegraphics[scale=0.06]{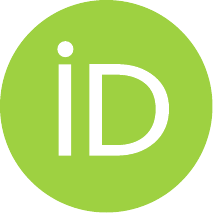}\hspace{1mm}Ezequiel L\'opez-Rubio}\thanks{Corresponding author. ITIS Software. Universidad de M\'alaga. C/ Arquitecto Francisco Peñalosa 18, 29010, Málaga, Spain} \\
	Department of Computer Languages and Computer Science\\
    University of M\'alaga\\
    Bulevar Louis Pasteur, 35\\
    29071 M\'alaga, Spain \\
	\texttt{ezeqlr@lcc.uma.es} \\
}

\renewcommand{\shorttitle}{Grid Programs}
\renewcommand{\headeright}{A Preprint}

\hypersetup{
  pdftitle={Grid Programs: A Two-Dimensional, Variable-Free Model of Computation},
  pdfsubject={cs.PL, cs.CC, cs.FL},
  pdfauthor={Ezequiel L\'opez-Rubio},
  pdfkeywords={models of computation, two-dimensional programs, visual programming, Turing completeness, stack-based computation},
}

\begin{document}
\maketitle

\begin{abstract}
We introduce \emph{Grid Programs}, a novel model of computation in which programs are finite two-dimensional arrangements of instructions on an integer grid rather than linear sequences of statements. Three properties distinguish this model fundamentally from classical frameworks: (i) programs are planar structures through which an instruction pointer moves in the four cardinal directions; (ii) there are no syntax constraints—any assignment of instructions to grid cells constitutes a valid program; and (iii) the model uses no named variables or explicit memory addresses. Program state is maintained through a data stack, an address stack, and a circularly doubly linked list accessed via three named pointers. Control flow is achieved spatially, with branching encoded as perpendicular turns of the instruction pointer. The address stack stores triplets (cell row, cell column, direction), enabling precise restoration of both position and heading after branches, loops, and function calls. We give a formal operational semantics, present a representative instruction set covering arithmetic, control flow, and linked-list manipulation, and work through several detailed examples, including an absolute-value function, a factorial computation, a linear-search algorithm, a string-reversal program, and a while-loop summation. We establish that Grid Programs are Turing-complete by simulating an arbitrary register machine, and we discuss their relationship to prior two-dimensional languages such as Befunge and Funge-98, to stack-based languages such as Forth and PostScript, and to dataflow and spatial computation models. Grid Programs offer a fresh vantage point for exploring the design space of computation, with potential applications in visual programming environments, cellular-automaton-inspired hardware, and obfuscation-resistant code.
\end{abstract}

\keywords{models of computation \and two-dimensional programs \and Turing completeness \and stack-based computation \and visual programming \and spatial instruction pointer \and function calls \and control structure nesting}

\section{Introduction}

Most programming languages, from assembly to Python, share a fundamental structural assumption: a program is a \emph{linear} sequence of instructions, and the flow of control is essentially one-dimensional. Even when execution branches or loops, the program text is a list that is read left-to-right and top-to-bottom. This linearity is so ingrained that it is rarely questioned. Yet there is no theoretical necessity for it: the underlying mathematics of computation do not require linearity.

In this paper we propose a genuinely two-dimensional model of computation that we call \emph{Grid Programs}. A Grid Program is a finite subset of the integer plane $\ZZ^2$, with each cell labeled by an instruction from a fixed finite set. Execution proceeds by moving an \emph{instruction pointer} (IP) across the grid in one of the four cardinal directions (left, right, up, down). Control-flow instructions change the direction of travel, branch to perpendicular paths, or loop back to earlier cells—all without any named labels or explicit jump addresses. Variables do not exist in the model; instead, working memory is provided by a data stack (DS), an address stack (AS), and a \emph{circularly doubly linked list} (CDLL) accessed through three named pointers.

The model has three properties that are individually known but have not previously been combined:

\begin{enumerate}[label=(\alph*)]
  \item \textbf{Planar programs.} Instructions occupy cells of the two-dimensional integer lattice. The IP moves in the plane rather than along a line.
  \item \textbf{No syntax.} Every function $f: \Dom \to \Inst$ from a finite subset $\Dom \subset \ZZ^2$ to the instruction set is a valid Grid Program. There are no parse errors, no reserved keywords that must appear in special positions, and no grammatical rules.
  \item \textbf{No variables.} The model offers neither named variables nor explicit memory addresses. All data manipulation goes through the DS, the AS, and the CDLL.
\end{enumerate}

\paragraph{Contributions.}
The main contributions of this paper are:
\begin{itemize}
  \item A formal definition of Grid Programs, including operational semantics for the complete instruction set (Section~\ref{sec:formal}).
  \item Five detailed worked examples that demonstrate how familiar algorithms—absolute value, sum of integers, factorial, linear search, and string reversal—are expressed as grid programs (Section~\ref{sec:examples}).
  \item A formal definition of the \instr{K} (call) instruction, which implements structured function calls via the address stack (Section~\ref{sec:formal}).
  \item A discussion of arbitrary nesting of control structures (Section~\ref{sec:nesting}).
  \item A proof that Grid Programs are Turing-complete (Section~\ref{sec:turing}).
  \item A comparative survey of related two-dimensional, stack-based, and unconventional computation models (Section~\ref{sec:related}).
  \item A discussion of expressiveness, design variants, and open problems (Section~\ref{sec:discussion}).
\end{itemize}

\paragraph{Organisation.}
Section~\ref{sec:overview} provides an informal overview of the model. Section~\ref{sec:formal} gives the formal definition and operational semantics. Section~\ref{sec:examples} presents five worked examples with step-by-step traces. Section~\ref{sec:nesting} shows that all control structures nest arbitrarily. Section~\ref{sec:turing} proves Turing completeness. Section~\ref{sec:related} surveys related work. Section~\ref{sec:discussion} discusses extensions and open questions, and Section~\ref{sec:conclusion} concludes.

\section{Informal Overview}
\label{sec:overview}

\subsection{Programs as Grids}

Imagine a sheet of graph paper. Each square cell may hold a single instruction symbol. The program is a finite populated region of that grid, which need not be connected. Execution begins at the cell labelled $(0,0)$, with the IP initially pointing \emph{upward}. At each step the IP reads the instruction in its current cell, performs the associated action, and advances—unless the instruction explicitly redirects it.

Figure~\ref{fig:overview_concept} illustrates the idea. The grid on the left shows a program occupying a roughly $L$-shaped region; the red arrow indicates the IP's current position and direction. As the IP moves, it may encounter \instr{T} (turn) instructions that rotate its heading by $90^\circ$, $180^\circ$, or $270^\circ$; \instr{F} (if) instructions that branch perpendicularly; or \instr{R}/\instr{U} instructions that implement a structured repeat/until loop.

\begin{figure}[ht]
  \centering
  \includegraphics[width=15cm]{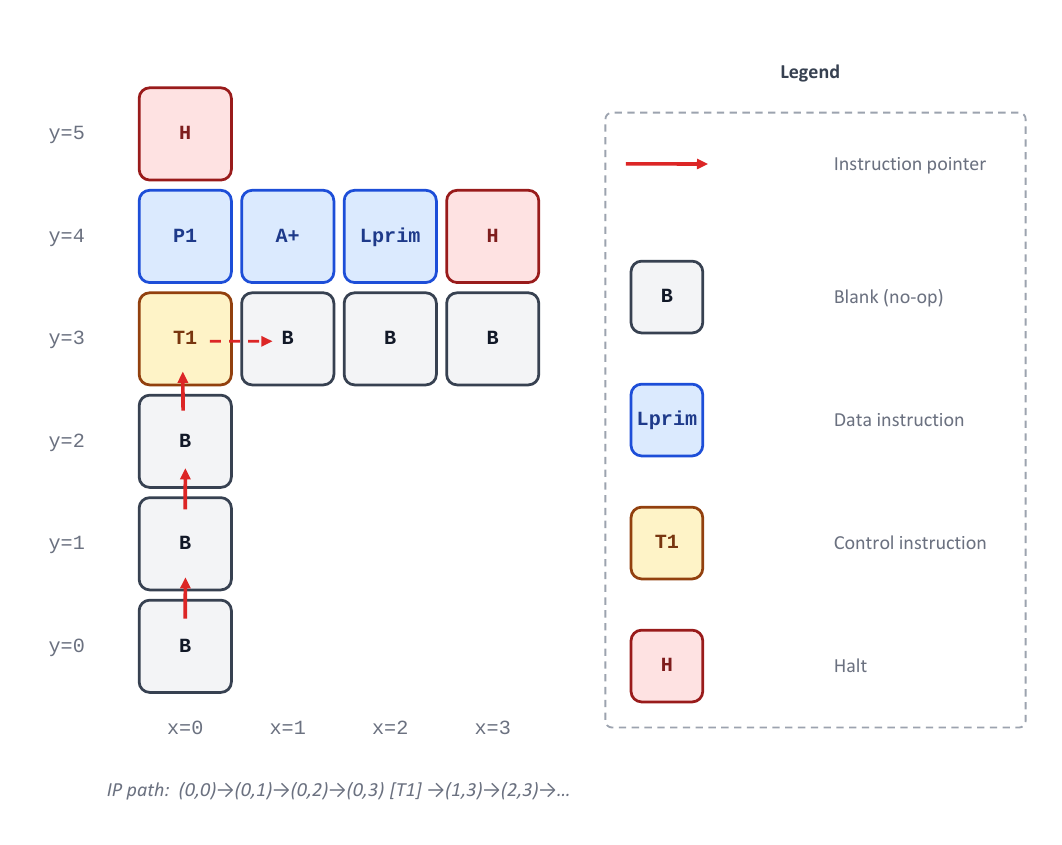}
  \caption{Conceptual view of a Grid Program. Each labeled cell holds one instruction. The red arrow indicates the instruction pointer and its current direction of travel. The IP started at $(0,0)$ (bottom-left of the occupied region), moving upward, executed several \instr{B} (blank, no-op) instructions, then encountered a \instr{T1} (turn right) instruction and is now moving rightward.}
  \label{fig:overview_concept}
\end{figure}

\subsection{Data Model}

The program state carries four data structures:
\begin{description}
  \item[\textbf{Data Stack (DS).}] A standard last-in-first-out stack of values from $\Val = \mathbb{B} \cup \ZZ \cup \RR \cup \Sigma^*$, covering Boolean, integer, real, and string data. Arithmetic and logic operations consume their operands from DS and push results back onto DS, in postfix (reverse-Polish) fashion.
  \item[\textbf{Address Stack (AS).}] A stack of \emph{triplets} $(x, y, d) \in \ZZ^2 \times \mathcal{C}$, where $(x,y)$ is a grid cell address and $d \in \mathcal{C}$ is a direction. Every control-flow instruction that saves a position also saves the current direction; every pop restores both address and direction simultaneously. This ensures that after a branch, loop, or function return, the IP resumes not only at the correct cell but also moving in the correct direction.
  \item[\textbf{Circularly Doubly Linked List (CDLL).}] An unbounded, circular, doubly linked list of values, providing random-access-like storage without named addresses. Three named pointers—$\Prim$, $\Sec$, and $\Ter$—each point to some node of the CDLL. Instructions load and store values through these pointers, insert and delete nodes, move a pointer to the location pointed to by another pointer, and copy values between the pointed-to nodes.
\end{description}

The combination of DS (for expression evaluation), AS (for control flow bookkeeping), and CDLL (for structured storage) gives the model the expressive power of an arbitrary Turing machine, as we prove in Section~\ref{sec:turing}.

\subsection{Control Flow by Direction}

The most distinctive feature of Grid Programs is how control flow is represented spatially. In a conventional language, an \texttt{if} statement contains two syntactic sub-sequences: the \texttt{then}-branch and the \texttt{else}-branch. In a Grid Program, the \instr{F} instruction sends the IP along a \emph{perpendicular} path depending on whether the top of DS is true or false:

\begin{itemize}
  \item If the condition is \textit{true}, the IP turns $90^\circ$ anticlockwise (e.g., if it was heading right, it now heads upward) and enters the \texttt{then}-branch.
  \item If the condition is \textit{false}, the IP turns $90^\circ$ clockwise (e.g., if heading right, it now heads downward) and enters the \texttt{else}-branch.
\end{itemize}

Both branches terminate with an \instr{E} (end) instruction, which pops the saved address from AS and resumes execution after the branching point. Loops are handled similarly: \instr{R} (repeat) saves the next cell's address on AS and begins the loop body; \instr{U} (until) checks the top of DS and either jumps back to the saved address (condition false) or advances and pops the address (condition true).

\paragraph{While loops (\instr{W}).}
The instruction \instr{W} (While) implements a pre-tested loop: it first checks a condition on the top of DS, then either exits or enters the loop body. Specifically:
\begin{enumerate}[label=(\roman*)]
  \item Pop $c$ from DS.
  \item \textbf{If $c = 0$ (false):} the loop condition is not met; advance the IP one step in $\dir$ (exiting the loop body region), leaving the body unexecuted.
  \item \textbf{If $c \neq 0$ (true):} push $(\IP, \dir)$---the address and direction of the \instr{W} cell itself---onto AS; rotate $\dir$ clockwise by $90^\circ$; move the IP to the first cell of the loop body (which lies in the new direction).
\end{enumerate}
After the loop body, an \instr{E} (end) instruction pops $(\IP_W, d_W)$ from AS and restores both position and direction, causing the IP to arrive back at the \instr{W} cell. The condition is re-evaluated, and the loop continues or terminates.

This design differs from the \instr{R}/\instr{U} (repeat\ldots until) construct in two ways: (i) the condition is tested \emph{before} each iteration rather than after, and (ii) a zero value exits the loop, making the condition a ``continue while non-zero'' test. The body of a \instr{W} loop is entered only if the condition is non-zero; the body of \instr{R}/\instr{U} is always executed at least once.

\paragraph{Function calls (\instr{K}).}
The instruction \instr{K} (Kall) performs a structured \emph{function call}. Executing \instr{K} does three things in sequence:
\begin{enumerate}[label=(\roman*)]
  \item push $\text{next}(\IP,\dir)$---the address of the cell immediately following \instr{K} in the current direction---onto AS (this is the \emph{return address});
  \item pop three values $d$, $y$, and $x$ from DS (top-to-bottom: $d$ on top, then $y$, then $x$), where $d \in \mathcal{C}$ is the entry direction of the function and $(x,y) \in \ZZ^2$ is the target address; set $\dir \leftarrow d$;
  \item move the IP to $(x, y)$.
\end{enumerate}
The called function body finishes with \instr{E} (end), which pops the return address from AS and resumes execution at the cell after \instr{K}. Because AS is a stack, calls may be nested to arbitrary depth and may be \emph{recursive}: each active call frame occupies exactly one entry on AS. The caller places the target coordinates and entry direction on DS immediately before \instr{K} (order on DS, bottom to top: $x$, $y$, $d$); the callee does not need to know where it was called from. The stored direction in AS ensures that when \instr{E} returns, the IP resumes moving in the direction the caller was travelling, not in the callee's entry direction.

This spatial encoding means that a single grid may contain multiple interleaved control structures simply by laying them out in orthogonal directions. Sequences of blank and turn cells may be employed as required to separate the control structures. The two-dimensional medium naturally separates control paths that would be tangled in a one-dimensional text.

\section{Formal Definition}
\label{sec:formal}

\subsection{Grid Programs}

\begin{definition}[Domain]
\label{def:domain}
A \emph{domain} is a finite subset $\Dom \subseteq \ZZ^2$ that contains the origin $(0,0)$. The domain need not be connected; cells that are not reachable by the IP during execution are simply never visited.
\end{definition}

\begin{definition}[Instruction set]
\label{def:instrset}
The \emph{instruction set} $\Inst$ is the finite set listed in Table~\ref{tab:instructions}.
\end{definition}

\begin{definition}[Grid Program]
\label{def:gridprogram}
A \emph{Grid Program} is a pair $\GP = (\Dom, f)$ where $\Dom$ is a domain and $f: \Dom \to \Inst$ is an instruction function that assigns an instruction to each cell of $\Dom$.
\end{definition}

No structural or syntactic constraint is imposed on $f$ beyond its type. Every such pair is a valid program.

\subsection{Values and State}

\begin{definition}[Value set]
The \emph{value set} is $\Val = \mathbb{B} \cup \ZZ \cup \RR \cup \Sigma^*$, where $\mathbb{B} = \{\mathsf{True}, \mathsf{False}\}$, $\ZZ$ is the set of integers, $\RR$ is the set of reals, and $\Sigma^*$ is the set of finite strings over a fixed alphabet $\Sigma$.
\end{definition}

A \emph{direction} is an element of $\mathcal{C} = \{\mathsf{left}, \mathsf{right}, \mathsf{up}, \mathsf{down}\}$. We define the unit step associated with a direction:
\[
  \delta(\mathsf{right}) = (1,0),\quad \delta(\mathsf{left}) = (-1,0),\quad
  \delta(\mathsf{up}) = (0,1),\quad \delta(\mathsf{down}) = (0,-1).
\]
The clockwise rotation operator $\rho$ and anticlockwise rotation $\rho^{-1}$ act on $\mathcal{C}$:
\[
  \rho: \mathsf{up}\mapsto\mathsf{right}\mapsto\mathsf{down}\mapsto\mathsf{left}\mapsto\mathsf{up}.
\]

\begin{definition}[Program state]
\label{def:state}
A \emph{program state} is a tuple
\[
  s = (\IP,\, \dir,\, \AS,\, \DS,\, \CDLL,\, \Prim,\, \Sec,\, \Ter)
\]
where:
\begin{itemize}
  \item $\IP \in \ZZ^2$ is the current instruction-pointer position,
  \item $\dir \in \mathcal{C}$ is the current direction,
  \item $\AS$ is a stack of triplets $(x, y, d) \in \ZZ^2 \times \mathcal{C}$,
  \item $\DS$ is a stack of elements of $\Val$,
  \item $\CDLL$ is a finite, non-empty, circularly doubly linked list of elements of $\Val$, and
  \item $\Prim, \Sec, \Ter$ are references to (not necessarily distinct) nodes of $\CDLL$.
\end{itemize}
\end{definition}

\begin{definition}[Initial state]
\label{def:initstate}
Given a Grid Program $\GP$ and an input sequence $(a_1, \ldots, a_n) \in \Val^*$, the \emph{initial state} is
\[
  s_0 = \bigl((0,0),\; \mathsf{up},\; \varnothing,\; \varnothing,\; [a_1, a_2, \ldots, a_n],\; \Prim{=}a_1,\; \Sec{=}a_1,\; \Ter{=}a_1\bigr)
\]
where $\varnothing$ denotes an empty stack, and all three CDLL pointers are initialised to the first node $a_1$.
\end{definition}

\subsection{Instruction Set and Operational Semantics}
\label{sec:semantics}

Table~\ref{tab:instructions} lists every instruction in $\Inst$. We write $\text{top}(\DS)$ for the top element of the data stack and $\text{next}(\IP, \dir) = \IP + \delta(\dir)$ for the next cell in the current direction.

\begin{table}[ht]
\centering
\caption{Complete instruction set for Grid Programs. The column \emph{Mnemonic} shows the symbol written in the grid cell.}
\label{tab:instructions}
\renewcommand{\arraystretch}{1.25}
\begin{tabularx}{\textwidth}{lllX}
\toprule
\textbf{Mnemonic} & \textbf{Name} & \textbf{Category} & \textbf{Description} \\
\midrule
\instr{B} & Blank & Control & No operation; IP advances by one step in direction $\dir$. \\
\instr{W} & While & Control & Pop $c$ from DS. If $c = 0$ (false): advance IP one step in $\dir$. If $c \neq 0$ (true): push $(\IP, \dir)$ onto AS; rotate $\dir$ clockwise by $90^\circ$; move IP to $\text{next}(\IP, \dir_{\text{new}})$. \\
\instr{H} & Halt & Control & Execution terminates; top of DS (if any) is the output. \\
\instr{R} & Repeat & Control & Push $(\text{next}(\IP,\dir), \dir)$ onto AS; move IP to that cell. \\
\instr{U} & Until & Control & If DS non-empty and $\text{top}(\DS)$ is falsy: move IP to $\text{top}(\AS)$ (do not pop). Otherwise: advance IP one step in $\dir$; if AS non-empty, pop AS. \\
\instr{F} & If & Control & Push $(\text{next}(\IP,\dir), \dir)$ onto AS. If DS non-empty and $\text{top}(\DS)$ truthy: move IP to $\text{next}(\IP,\rho^{-1}(\dir))$ (anticlockwise turn). Otherwise: move IP to $\text{next}(\IP,\rho(\dir))$ (clockwise turn). \\
\instr{E} & End/Return & Control & Pop $(a, d)$ from AS (if non-empty); set $\dir \leftarrow d$; move IP to $a$. \\
\instr{K} & Kall & Control & Push $(\text{next}(\IP,\dir), \dir)$ onto AS (return address + current direction). Pop $d$, $y$, then $x$ from DS; set $\dir \leftarrow d$; move IP to $(x, y)$. \\
\instr{Ao} & Arith/Logic & Data & Perform operation $o$ on top element(s) of DS; pop operands, push result. \\
\instr{Ly} & Load & Data & Push value at CDLL node pointed by $y \in \{\Prim,\Sec,\Ter\}$ onto DS. \\
\instr{Sy} & Store & Data & Pop top of DS; write it to the CDLL node pointed by $y$. \\
\instr{Iy} & Insert & Data & Insert a new node with value $0$ immediately after the node pointed by $y$; $y$ is updated to point to the new node. \\
\instr{Dy} & Delete & Data & If $|\CDLL|>1$: delete the node pointed by $y$; $y$ advances to the next node. \\
\instr{Mxy} & Move pointer & Data & Set pointer $x$ to point to the same node as pointer $y$. \\
\instr{Cxy} & Copy & Data & Write the value at pointer $x$ to the node pointed by pointer $y$. \\
\instr{Nxy} & Next & Data & Advance pointer $x$ by one step in CDLL traversal direction; $y \in \{+,-\}$ selects forward/backward. \\
\instr{Pk} & Push const & Data & Push constant $k \in \{0, 1, e, \pi\}$ onto DS. \\
\instr{X} & Pop & Data & If DS non-empty: pop the top element. \\
\instr{Tz} & Turn & Control & Rotate $\dir$ clockwise by $z \times 90^\circ$, where $z \in \{1,2,3\}$; advance IP one step in new direction. \\
\bottomrule
\end{tabularx}
\end{table}

\paragraph{Arithmetic/logic operations.}
The \instr{Ao} instruction parameterises over a standard set of operations $o$, which includes:
\[
  \{+,\,-,\,\times,\,\div,\,\mathsf{mod},\,\mathsf{neg},\,\mathsf{abs},\,\mathsf{pow},\,\mathsf{sqrt},\,\mathsf{floor},\,\mathsf{ceil},\,
   =,\,\neq,\,<,\,\leq,\,>,\,\geq,\,\mathsf{and},\,\mathsf{or},\,\mathsf{not},\,\mathsf{concat},\,\mathsf{len},\ldots\}
\]
Binary operations consume two values from DS (second-from-top is the left operand, top is the right operand) and push one result.

\paragraph{Advance function.}
Let $\text{adv}(s) = s'$ denote the state obtained by executing the instruction at $\IP$ in state $s$. The full transition relation is specified instruction by instruction in Table~\ref{tab:transitions}.

\begin{table}[ht]
\centering
\caption{Operational semantics: selected transition rules. $\ip' = \IP + \delta(\dir)$ is the default next cell. $\text{top}(s)$ denotes $\text{top}(\DS)$ in state $s$.}
\label{tab:transitions}
\renewcommand{\arraystretch}{1.3}
\begin{tabular}{lll}
\toprule
\textbf{Instruction} & \textbf{Precondition} & \textbf{Effect on state} \\
\midrule
\instr{B} & — & $\IP \leftarrow \ip'$ \\
\instr{W} & $\DS \neq \varnothing$, $c = \text{top}(\DS) = 0$ & pop $\DS$;\ $\IP \leftarrow \ip'$ \\
\instr{W} & $\DS \neq \varnothing$, $c = \text{top}(\DS) \neq 0$ & pop $\DS$;\ $\AS \leftarrow \push(\AS, (\IP, \dir))$;\ $\dir \leftarrow \rho(\dir)$;\ $\IP \leftarrow \IP+\delta(\dir) $ \\
\instr{H} & — & terminate \\
\instr{R} & — & $\AS \leftarrow \push(\AS, (\ip', \dir))$;\ $\IP \leftarrow \ip'$ \\
\instr{U} & $\DS \neq \varnothing$, $\text{top}(\DS)$ falsy & $(a,d) \leftarrow \text{top}(\AS)$;\ $\dir \leftarrow d$;\ $\IP \leftarrow a$ \\
\instr{U} & $\DS = \varnothing$ or $\text{top}(\DS)$ truthy & $\IP \leftarrow \ip'$;\ $\AS \leftarrow \pop(\AS)$ \\
\instr{F} & $\DS \neq \varnothing$, $\text{top}(\DS)$ truthy & $\AS \leftarrow \push(\AS, (\ip', \dir))$;\ $\IP \leftarrow \IP + \delta(\rho^{-1}(\dir))$ \\
\instr{F} & $\DS = \varnothing$ or $\text{top}(\DS)$ falsy & $\AS \leftarrow \push(\AS, (\ip', \dir))$;\ $\IP \leftarrow \IP + \delta(\rho(\dir))$ \\
\instr{E} & $\AS \neq \varnothing$ & $(a,d) \leftarrow \text{top}(\AS)$;\ $\AS \leftarrow \pop(\AS)$;\ $\dir \leftarrow d$;\ $\IP \leftarrow a$ \\
\instr{K} & $|\DS| \geq 3$ & $\AS \leftarrow \push(\AS, (\ip^{\prime}, \dir))$; pop $d$ then $y$ then $x$ from $\DS$; $\dir \leftarrow d$;\ $\IP \leftarrow (x, y)$ \\
\instr{Ao} & operands on DS & pop operands; push $o(\text{operands})$;\ $\IP \leftarrow \ip'$ \\
\instr{Ly} & — & $\DS \leftarrow \push(\DS, \CDLL[y])$;\ $\IP \leftarrow \ip'$ \\
\instr{Sy} & $\DS \neq \varnothing$ & $\CDLL[y] \leftarrow \text{top}(\DS)$;\ $\DS \leftarrow \pop(\DS)$;\ $\IP \leftarrow \ip'$ \\
\instr{Pk} & — & $\DS \leftarrow \push(\DS, k)$;\ $\IP \leftarrow \ip'$ \\
\instr{X} & $\DS \neq \varnothing$ & $\DS \leftarrow \pop(\DS)$;\ $\IP \leftarrow \ip'$ \\
\instr{Tz} & — & $\dir \leftarrow \rho^z(\dir)$;\ $\IP \leftarrow \IP + \delta(\rho^z(\dir))$ \\
\bottomrule
\end{tabular}
\end{table}

\begin{remark}[Robustness to undefined states]
If the IP steps off the boundary of $\Dom$ (i.e., $\IP + \delta(\dir) \notin \Dom$), the program is considered to have halted abnormally. A well-designed grid program ensures the IP always remains within $\Dom$, typically by enclosing active regions with \instr{B} or \instr{H} cells that prevent accidental escape.
\end{remark}

\section{Worked Examples}
\label{sec:examples}

We present five programs of increasing complexity. For each program we (i)~display the grid with annotated cell coordinates, (ii)~explain the layout and control-flow strategy, and (iii)~trace the execution on a representative input.

Throughout this section, the IP begins at $(0,0)$ pointing \emph{upward} ($\dir=\mathsf{up}$), i.e., toward increasing~$y$, unless otherwise stated. Notation: $\ip' = \text{next}(\IP,\dir)$ is the cell immediately ahead of the IP in its current direction.

\subsection{Absolute Value}
\label{sec:abs}

\begin{example}[Absolute value]
\label{ex:abs}
Compute $|n|$ for an integer $n$ given on top of DS. The IP starts at $(0,0)$ moving upward.
\end{example}

\begin{figure}[ht]
\centering
\includegraphics[width=15cm]{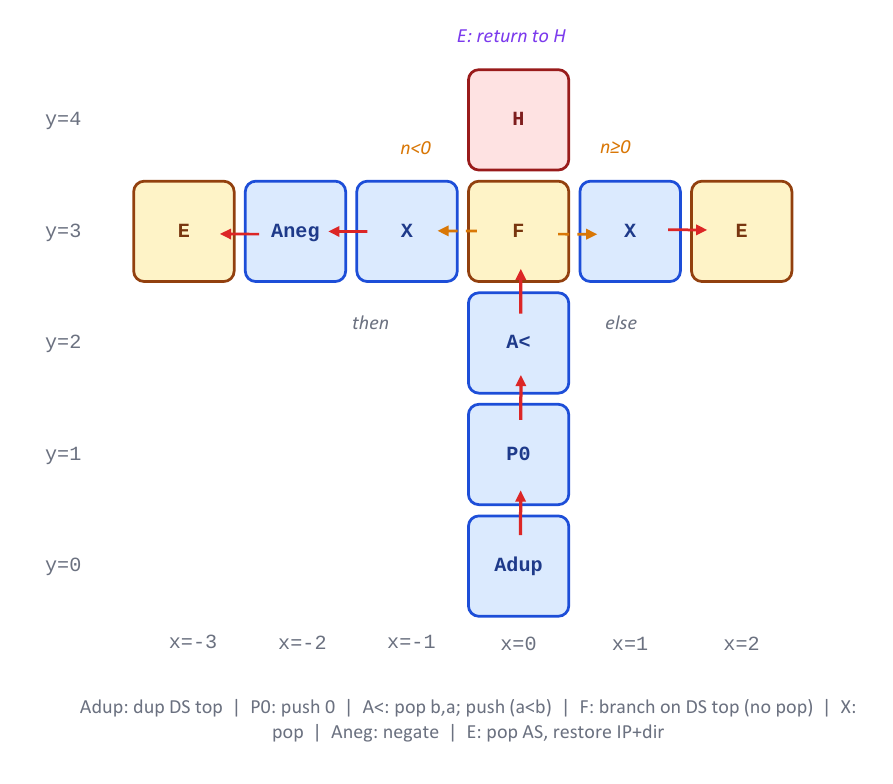}
\caption{Grid Program for $|n|$. Main spine: column $x=0$, moving upward.
  \instr{F} at $(0,3)$ branches: \emph{true} ($n<0$) turns left (CCW from up) to
  the negation arm; \emph{false} ($n\geq 0$) turns right (CW from up) and returns
  directly.  Both arms end with \instr{E}, which restores the IP to $(0,4)$.}
\label{fig:abs}
\end{figure}

\paragraph{Grid layout.}

\[
\begin{array}{r|ccccc}
y & x{=}{-3} & x{=}{-2} & x{=}{-1} & x{=}0 & x{=}1 \quad x{=}2 \\ \hline
4 & & & & \instr{H}   & \\
3 & \instr{E} & \instr{Aneg} & \instr{X} & \instr{F} & \instr{X} \quad \instr{E} \\
2 & & & & \instr{A<}  & \\
1 & & & & \instr{P0}  & \\
0 & & & & \instr{Adup}& \\
\end{array}
\]

\noindent The main spine runs along $x=0$ from $y=0$ to $y=4$.  \instr{F} at $(0,3)$
does \emph{not} pop its condition; each branch begins with \instr{X} to discard it.

\begin{itemize}
  \item \textbf{True branch} ($n<0$, IP turns CCW = left):
        $(- 1,3)$~\instr{X}, $(-2,3)$~\instr{Aneg}, $(-3,3)$~\instr{E}.
  \item \textbf{False branch} ($n\geq 0$, IP turns CW = right):
        $(1,3)$~\instr{X}, $(2,3)$~\instr{E}.
\end{itemize}
Both \instr{E} cells pop $(0,4,\mathsf{up})$ from AS and land on \instr{H}.

\paragraph{Execution trace for $n=-5$.}
Initial DS: $[-5]$, AS: $[]$, $\dir=\mathsf{up}$, $\IP=(0,0)$.
\begin{enumerate}
  \item $(0,0)$~\instr{Adup}: DS$=[-5,-5]$.
  \item $(0,1)$~\instr{P0}: DS$=[-5,-5,0]$.
  \item $(0,2)$~\instr{A<}: pop $0,-5$; push $1$; DS$=[-5,1]$.
  \item $(0,3)$~\instr{F}: condition $=1$ (true); push $(0,4,\mathsf{up})$ to AS; turn CCW $\Rightarrow$ $\dir=\mathsf{left}$; IP$\to(-1,3)$.
  \item $(-1,3)$~\instr{X}: pop $1$; DS$=[-5]$.
  \item $(-2,3)$~\instr{Aneg}: DS$=[5]$.
  \item $(-3,3)$~\instr{E}: pop AS; $\dir\leftarrow\mathsf{up}$; IP$\to(0,4)$.
  \item $(0,4)$~\instr{H}: output $=5$.~\checkmark
\end{enumerate}

\paragraph{Execution trace for $n=3$.}
Steps 1--3 give DS$=[3,0]$; \instr{A<} pushes $0$ (false).
\instr{F}: condition $=0$ (false); push $(0,4,\mathsf{up})$; turn CW $\Rightarrow$ $\dir=\mathsf{right}$; IP$\to(1,3)$.
$(1,3)$~\instr{X}: DS$=[3]$.  $(2,3)$~\instr{E}: IP$\to(0,4)$.  \instr{H}: output $=3$.~\checkmark

\subsection{Factorial}
\label{sec:factorial}

\begin{example}[Factorial via \instr{W} while loop]
\label{ex:factorial}
Compute $n!$ for a non-negative integer $n$.  The CDLL stores: node~0 ($\Prim$)~$=$ accumulator $\mathit{acc}=1$; node~1 ($\Sec$)~$=$ counter $c=n$.  DS is empty at start.  IP starts at $(0,0)$ moving upward.
\end{example}

\begin{figure}[ht]
\centering
\includegraphics[width=15cm]{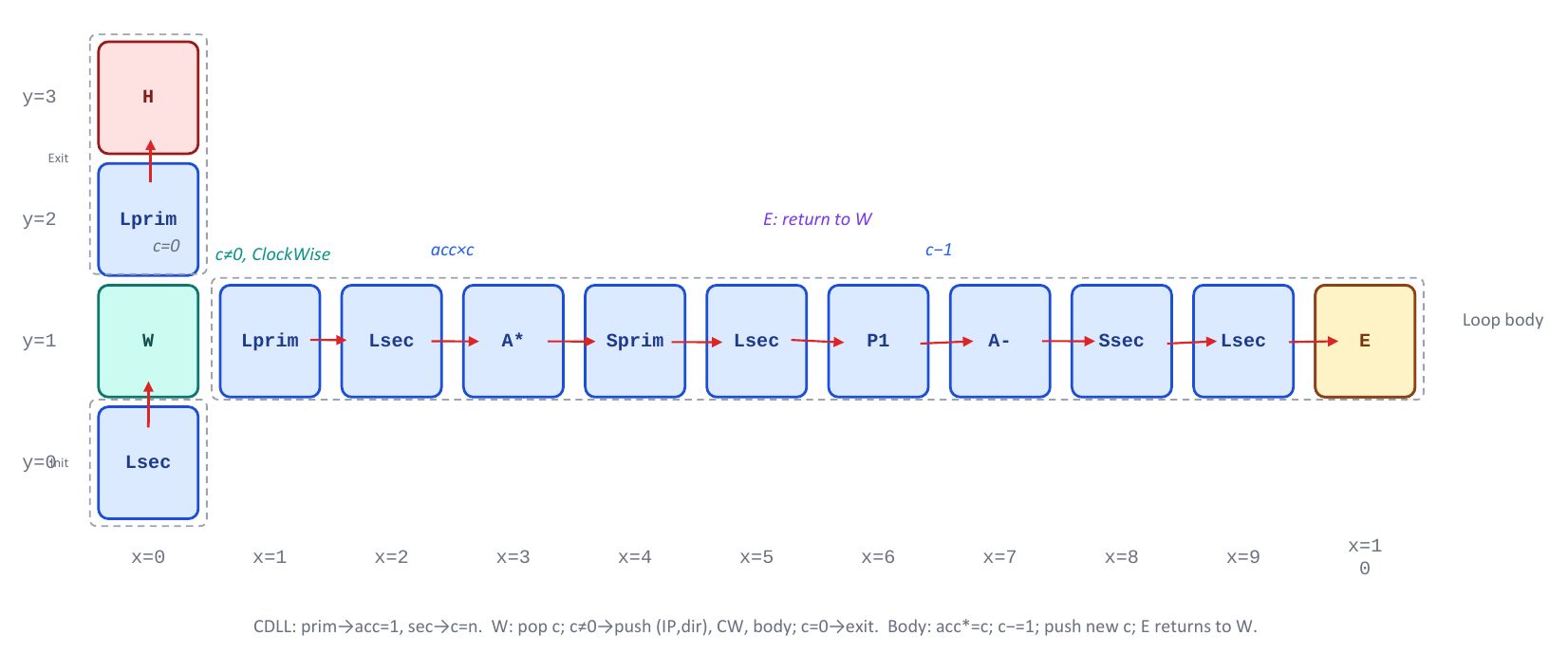}
\caption{Grid Program computing $n!$ using the \instr{W} (while) instruction.
  The main spine (column $x=0$, upward) reads the counter and enters the
  \instr{W} cell.  When the counter is non-zero, the IP turns right into the
  loop body; \instr{E} at the end of the body returns to \instr{W}.
  When the counter reaches zero, the IP exits upward, loads the accumulator,
  and halts.}
\label{fig:factorial}
\end{figure}

\paragraph{Grid layout.}

\[
\begin{array}{r|ccccccccccc}
y & x{=}0 & x{=}1 & x{=}2 & x{=}3 & x{=}4 & x{=}5 & x{=}6 & x{=}7 & x{=}8 & x{=}9 & x{=}10\\ \hline
3 & \instr{H}    & & & & & & & & & & \\
2 & \instr{Lprim}& & & & & & & & & & \\
1 & \instr{W}    & \instr{Lprim} & \instr{Lsec} & \instr{A*} & \instr{Sprim} & \instr{Lsec} & \instr{P1} & \instr{A-} & \instr{Ssec} & \instr{Lsec} & \instr{E} \\
0 & \instr{Lsec} & & & & & & & & & & \\
\end{array}
\]

\noindent The body row ($y=1$, $x=1\ldots10$) runs rightward. \instr{E} at $(10,1)$ pops
$(0,1,\mathsf{up})$ from AS and restores the IP to the \instr{W} cell, which then
re-evaluates the condition.

\paragraph{Execution trace for $n=3$.}
CDLL: $\Prim=1$, $\Sec=3$.  DS: $[]$.
\begin{enumerate}
  \item $(0,0)$~\instr{Lsec}: push $3$; DS$=[3]$.
  \item $(0,1)$~\instr{W}: pop $3\neq0$; push $(0,1,\mathsf{up})$; turn CW $\Rightarrow$ $\dir=\mathsf{right}$; IP$\to(1,1)$.
  \item \emph{Body iteration 1} ($c=3$): $\mathit{acc}\times c=1\times3=3$; $c\leftarrow2$. \instr{Lsec} pushes $2$; \instr{E} returns to $(0,1)$.
  \item \emph{Body iteration 2} ($c=2$): $3\times2=6$; $c\leftarrow1$. \instr{E} returns to $(0,1)$.
  \item \emph{Body iteration 3} ($c=1$): $6\times1=6$; $c\leftarrow0$. \instr{Lsec} pushes $0$; \instr{E} returns.
  \item $(0,1)$~\instr{W}: pop $0$; exit upward; IP$\to(0,2)$.
  \item $(0,2)$~\instr{Lprim}: push $6$; DS$=[6]$.
  \item $(0,3)$~\instr{H}: output $=6=3!$.~\checkmark
\end{enumerate}

\subsection{While Loop: Sum of First $n$ Integers}
\label{sec:while_example}

\begin{example}[Sum $1+2+\cdots+n$ using a while loop]
\label{ex:while_sum}
Compute $S=n(n+1)/2$ for $n\geq0$ using the \instr{W} instruction.
CDLL: node~0 ($\Prim$) $=$ counter $c$ (initialised to $n$); node~1 ($\Sec$) $=$ accumulator $S$ (initialised to $0$).  DS is empty at start.  IP starts at $(0,0)$ moving upward.
\end{example}

\begin{figure}[ht]
\centering
\includegraphics[width=15cm]{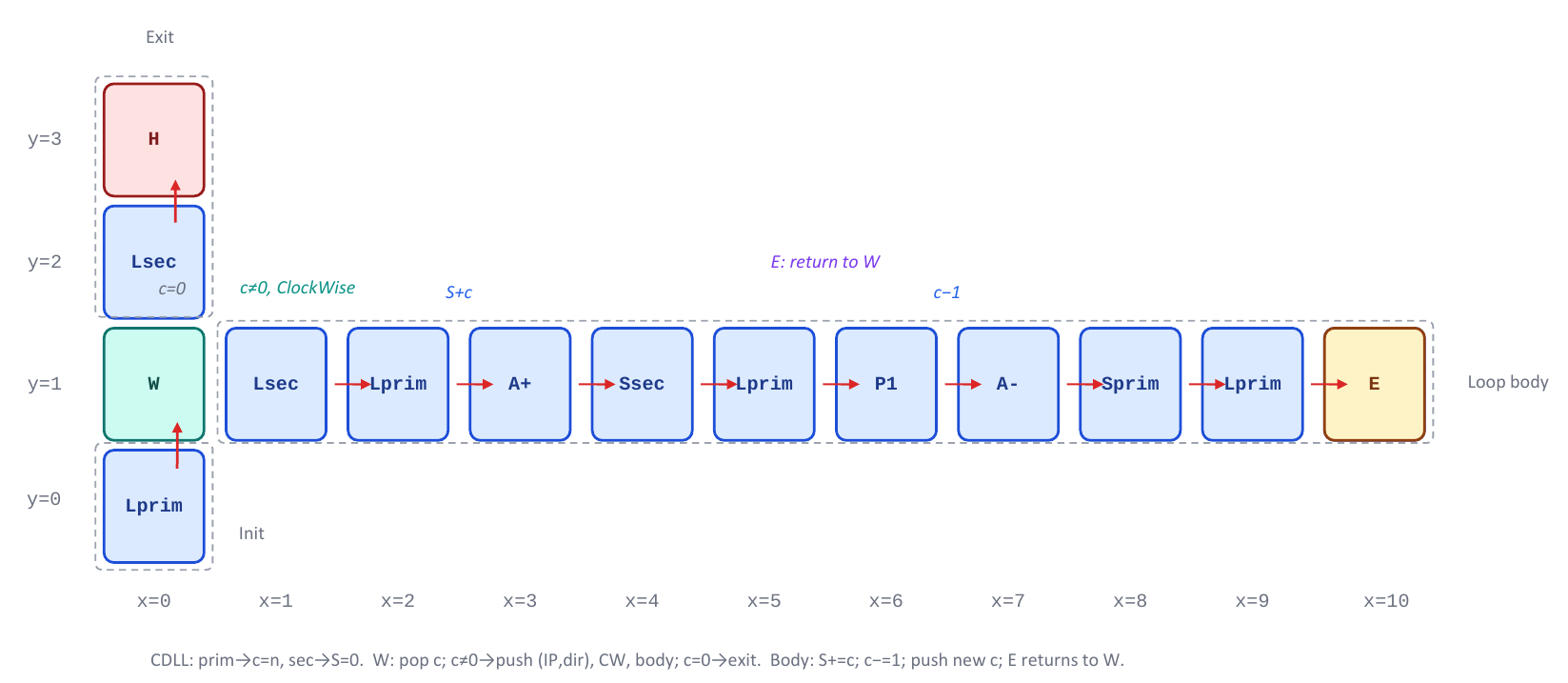}
\caption{Grid Program computing $\sum_{i=1}^{n}i$ with the \instr{W} (while) instruction (teal cell).
  The main spine (column $x=0$, upward) loads $c$ and enters \instr{W}.
  When $c\neq0$ the IP turns right and runs the body: $S\mathrel{+}=c$; $c\mathrel{-}=1$; push new $c$.
  \instr{E} returns to \instr{W}. When $c=0$ the IP exits upward, loads $S$, and halts.}
\label{fig:while}
\end{figure}

\paragraph{Grid layout.}

\[
\begin{array}{r|ccccccccccc}
y & x{=}0 & x{=}1 & x{=}2 & x{=}3 & x{=}4 & x{=}5 & x{=}6 & x{=}7 & x{=}8 & x{=}9 & x{=}10\\ \hline
3 & \instr{H}    & & & & & & & & & & \\
2 & \instr{Lsec} & & & & & & & & & & \\
1 & \instr{W}    & \instr{Lsec} & \instr{Lprim} & \instr{A+} & \instr{Ssec} & \instr{Lprim} & \instr{P1} & \instr{A-} & \instr{Sprim} & \instr{Lprim} & \instr{E}\\
0 & \instr{Lprim}& & & & & & & & & & \\
\end{array}
\]

\noindent \instr{W} at $(0,1)$ pops the condition (pushed by \instr{Lprim} at $(0,0)$ on first entry, and by the body's final \instr{Lprim} at $(9,1)$ on re-entry). When $c\neq0$: push $(0,1,\mathsf{up})$; turn CW $\Rightarrow$ right; enter body at $(1,1)$.  Body: $S\mathrel{+}=c$ (steps 1--4), $c\mathrel{-}=1$ (steps 5--8), push new $c$ (step 9), then \instr{E} at $(10,1)$ restores $(0,1,\mathsf{up})$.  When $c=0$: advance upward to $(0,2)$.

\paragraph{Execution trace for $n=3$.}
CDLL: $\Prim=3$, $\Sec=0$.
\begin{enumerate}
  \item $(0,0)$~\instr{Lprim}: push $3$; DS$=[3]$.
  \item $(0,1)$~\instr{W}: pop $3\neq0$; push $(0,1,\mathsf{up})$; turn right; IP$\to(1,1)$.
  \item \emph{Iteration 1} ($c=3$): $S\leftarrow0+3=3$; $c\leftarrow2$; push $2$; \instr{E} returns.
  \item \emph{Iteration 2} ($c=2$): $S\leftarrow3+2=5$; $c\leftarrow1$; push $1$; \instr{E} returns.
  \item \emph{Iteration 3} ($c=1$): $S\leftarrow5+1=6$; $c\leftarrow0$; push $0$; \instr{E} returns.
  \item $(0,1)$~\instr{W}: pop $0$; exit upward; IP$\to(0,2)$.
  \item $(0,2)$~\instr{Lsec}: push $6$.
  \item $(0,3)$~\instr{H}: output $=6=3\times4/2$.~\checkmark
\end{enumerate}

\subsection{Linear Search}
\label{sec:linsearch}

\begin{example}[Linear search]
\label{ex:linsearch}
Given a CDLL pre-loaded as $[t, a_0, a_1, \ldots, a_{k-1}]$—target $t$ in node~0 and elements $a_0,\ldots,a_{k-1}$ in nodes~$1,\ldots,k$—find the first index $i\in\{0,\ldots,k-1\}$ such that $a_i=t$, or $-1$ if not found. DS is empty on entry.
\end{example}

\begin{figure}[ht]
\centering
\includegraphics[width=15cm]{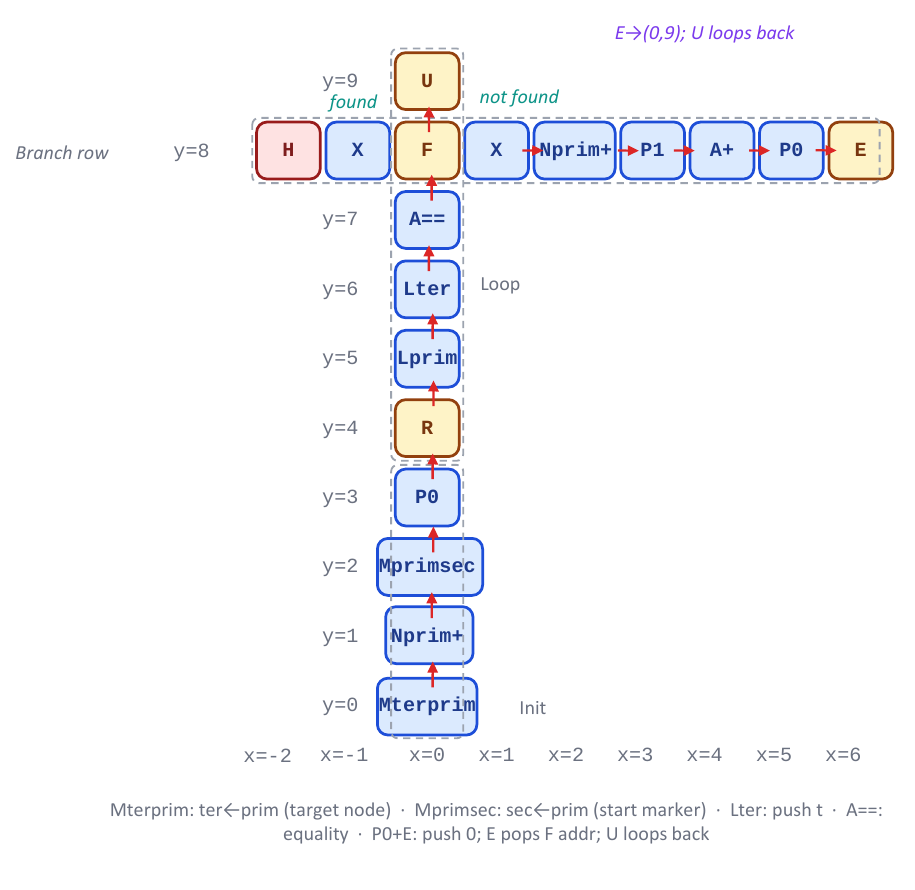}
\caption{Grid Program for linear search. The outer loop advances the CDLL pointer; the inner conditional checks equality and returns the current index or continues.}
\label{fig:linsearch}
\end{figure}

\paragraph{Strategy.}
The CDLL layout $[t, a_0, \ldots, a_{k-1}]$ puts the target $t$ permanently in node~0 (pointed to by $\Ter$ after init) so it is always accessible via \instr{Lter}. $\Prim$ starts at $a_0$ (the first element) and scans forward; $\Sec$ marks the start of the elements. A counter on DS tracks the current index, starting at $0$.

The loop uses \instr{R}/\instr{U}: each body iteration compares $a[\Prim]$ with $t$ via \instr{F}; on equality the found branch (\instr{H}) outputs the index. The not-found branch advances $\Prim$, increments the index, then pushes $0$ and executes \instr{E} to exit the \instr{F} branch, landing at the \instr{U} cell which loops back (condition $0$) to restart the body.

\emph{Note.} This grid has no termination condition for the not-found case: if $t$ is absent from the CDLL the loop runs indefinitely. A complete implementation would track the number of steps (e.g.\ by storing $k$ in a second CDLL node) and break after $k$ iterations, pushing $-1$.

\paragraph{Grid layout.}

The CDLL is pre-loaded as $[t, a_0, \ldots, a_{k-1}]$, with \texttt{prim}, \texttt{sec}, and \texttt{ter} all initially pointing to node~0 (the target node).  DS is empty on entry.

\[
\begin{array}{r|ccccccc}
y & x{=}{-2} & x{=}{-1} & x{=}0 & x{=}1 & x{=}2 & x{=}3 \cdots x{=}6\\ \hline
9 & & & \instr{U}     & & & \\
8 & \instr{H} & \instr{X} & \instr{F} & \instr{X} & \instr{Nprim+} & \instr{P1}\;\instr{A+}\;\instr{P0}\;\instr{E} \\
7 & & & \instr{A==}   & & & \\
6 & & & \instr{Lter}  & & & \\
5 & & & \instr{Lprim} & & & \\
4 & & & \instr{R}     & & & \\
3 & & & \instr{P0}    & & & \\
2 & & & \instr{Mprimsec}& & & \\
1 & & & \instr{Nprim+} & & & \\
0 & & & \instr{Mterprim}& & & \\
\end{array}
\]

\noindent \textbf{Init} (column $x=0$, $y=0\ldots3$, going upward):
\begin{itemize}
  \item $y=0$: \instr{Mterprim} — set $\Ter \leftarrow \Prim$ (so $\Ter$ permanently points to the target node).
  \item $y=1$: \instr{Nprim+}   — advance $\Prim$ to $a_0$ (first element).
  \item $y=2$: \instr{Mprimsec} — set $\Sec \leftarrow \Prim$ (marks start of elements).
  \item $y=3$: \instr{P0}       — push \texttt{index}$=0$.
\end{itemize}

\noindent \textbf{Loop} (\instr{R}/\instr{U}, column $x=0$, $y=4\ldots9$):
\begin{itemize}
  \item $y=4$: \instr{R} — push $(0,5,\mathsf{up})$ onto AS; body starts at $(0,5)$.
  \item $y=5\ldots7$: push element (\instr{Lprim}), push target (\instr{Lter}), compare (\instr{A==}).
  \item $y=8$: \instr{F} — condition $= (a_i = t)$.  \emph{Found} (true): IP turns CCW (left) to $(-1,8)$, $(-2,8)$.  \emph{Not found} (false): IP turns CW (right) to $(1,8)\ldots(6,8)$.
\end{itemize}

\noindent \textbf{Found branch} (left from \instr{F} at $(0,8)$ going up, so IP turns left):
$(- 1,8)$~\instr{X} pops the condition; $(-2,8)$~\instr{H} outputs \texttt{index}.

\noindent \textbf{Not-found branch} (right from \instr{F}):
\instr{X} (pop condition), \instr{Nprim+} (advance scan pointer), \instr{P1}~\instr{A+} (increment index), \instr{P0} (push $0$ as loop-back condition for \instr{U}), \instr{E} (pop \instr{F}'s saved address; jump to $(0,9)$). Then \instr{U} at $(0,9)$: sees $c=0$, peeks \instr{R}'s saved address $(0,5)$, and jumps back to restart the body. The \instr{U} condition is always $0$ (loop back) during normal search; the only exit is the found branch via \instr{H}.

\subsection{String Reversal}
\label{sec:strrev}

\begin{example}[String reversal]
\label{ex:strrev}
Given a string $s=c_0c_1\cdots c_{k-1}$ of $k$ characters loaded into the CDLL (one character per node, 0-indexed), produce the reversed string $c_{k-1}\cdots c_1c_0$.
\end{example}

\begin{figure}[ht]
\centering
\includegraphics[width=15cm]{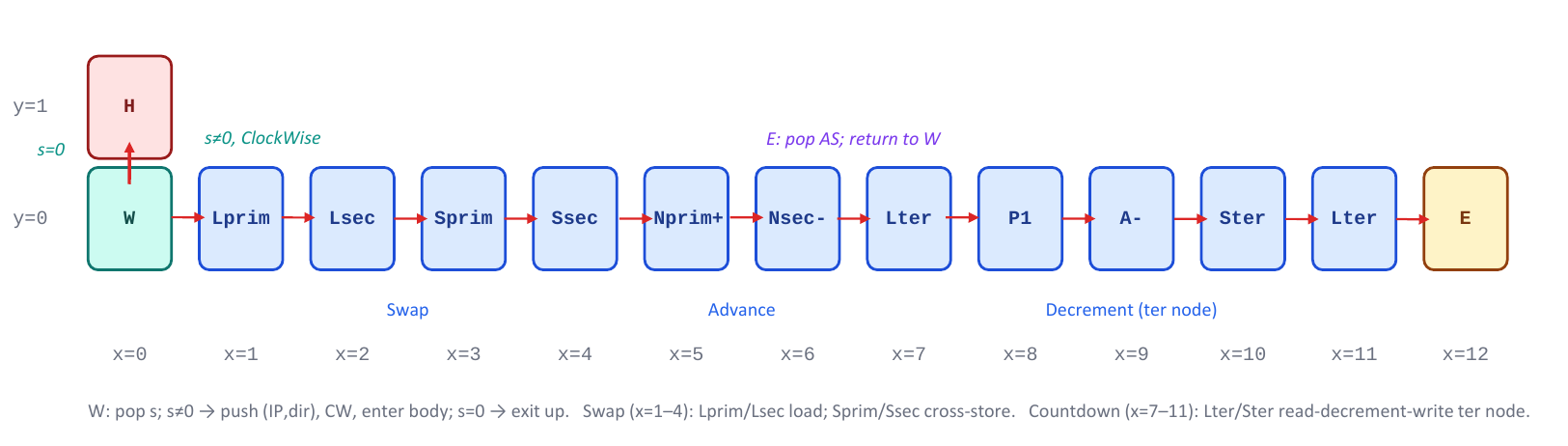}
\caption{Grid Program for string reversal using the CDLL. The program walks two pointers from opposite ends of the list, swapping characters until they meet.}
\label{fig:strrev}
\end{figure}

\paragraph{Strategy.}
Place $\Prim$ at $c_0$ (node~0) and $\Sec$ at $c_{k-1}$ (node~$k-1$). In each swap iteration: push $c[\Prim]$ and $c[\Sec]$ onto DS and cross-store; advance $\Prim$ forward (toward $c_{k-1}$) and $\Sec$ backward (toward $c_0$); decrement the swap counter stored in $\Ter$. Halt when the counter reaches zero after $\lfloor k/2\rfloor$ swaps.

\paragraph{Grid layout.}

The CDLL is pre-loaded as $[c_0, c_1, \ldots, c_{k-1}, s]$ where $s=\lfloor k/2\rfloor$ is the swap count.
$\Prim$ points to $c_0$ (node~0), $\Sec$ points to $c_{k-1}$ (node~$k-1$), and $\Ter$ points to the swap-count node (node~$k$).  DS contains $[s]$ on entry.

\[
\begin{array}{r|ccccccccccccc}
y & x{=}0 & x{=}1 & x{=}2 & x{=}3 & x{=}4 & x{=}5 & x{=}6 & x{=}7 & x{=}8 & x{=}9 & x{=}10 & x{=}11 & x{=}12\\ \hline
1 & \instr{H}  & & & & & & & & & & & & \\
0 & \instr{W}  & \instr{Lprim} & \instr{Lsec} & \instr{Sprim} & \instr{Ssec} & \instr{Nprim+} & \instr{Nsec-} & \instr{Lter} & \instr{P1} & \instr{A-} & \instr{Ster} & \instr{Lter} & \instr{E} \\
\end{array}
\]

\noindent The \instr{W} cell at $(0,0)$ acts as the loop head. On each iteration:
\begin{enumerate}
  \item \instr{W} pops the countdown $s$ from DS. If $s=0$: exit upward to \instr{H}. If $s\neq0$: push $(0,0,\mathsf{up})$ onto AS; turn CW (up$\to$right); enter body at $(1,0)$.
  \item \textbf{Swap} ($x=1\ldots4$): \instr{Lprim} pushes $c[\Prim]$; \instr{Lsec} pushes $c[\Sec]$; \instr{Sprim} stores old $c[\Sec]$ at $\Prim$; \instr{Ssec} stores old $c[\Prim]$ at $\Sec$.
  \item \textbf{Advance pointers} ($x=5\ldots6$): \instr{Nprim+} (left pointer moves right); \instr{Nsec-} (right pointer moves left).
  \item \textbf{Decrement countdown} ($x=7\ldots10$): \instr{Lter} loads $s$ from the CDLL countdown node; \instr{P1}~\instr{A-} computes $s-1$; \instr{Ster} stores $s-1$ back.
  \item \textbf{Re-arm} ($x=11$): \instr{Lter} pushes the new $s-1$ onto DS (as condition for the next \instr{W}).
  \item \instr{E} at $x=12$ pops $(0,0,\mathsf{up})$ from AS; IP returns to \instr{W}.
\end{enumerate}

\noindent The key design point is that the countdown is stored in the CDLL (not only on DS), because \instr{W} \emph{pops} the condition before entering the body; the body therefore reloads the countdown from $\Ter$ and pushes the updated value for the next \instr{W} call.

\section{Arbitrary Nesting of Control Structures}
\label{sec:nesting}

One of the most important structural properties of Grid Programs is that \emph{all control structures—branches (\instr{F}/\instr{E}), pre-tested loops (\instr{W}/\instr{E}), post-tested loops (\instr{R}/\instr{U}), and function calls (\instr{K}/\instr{E})—can be nested to an arbitrary depth}, limited only by the size of the address stack AS. Figure \ref{fig:nesting_schematic} illustrates this fact.

\subsection*{Why nesting works}

The address stack AS acts as the single universal control stack. Each control structure, on entry, pushes exactly one triplet $(x, y, d)$ onto AS:
\begin{itemize}
  \item \instr{R} pushes the loop-body entry address $(x', y', d)$.
  \item \instr{W} (when condition is non-zero) pushes its own cell address $(x_W, y_W, d)$.
  \item \instr{F} pushes the resume address $(x_F', y_F', d)$.
  \item \instr{K} pushes the return address $(x_{\mathrm{ret}}, y_{\mathrm{ret}}, d)$.
\end{itemize}
Each corresponding exit instruction (\instr{U} on a true condition, \instr{E}) pops exactly one triplet. Since AS is an unbounded stack and each level contributes exactly one entry, the maximum depth of AS during execution equals the maximum nesting depth of active control structures. There is no architectural limit.

\subsection*{Spatial encoding}

In a one-dimensional language, nesting depth is encoded syntactically by bracket depth or indentation. In a Grid Program it is encoded \emph{geometrically}: each nested level occupies a different spatial direction or region of the grid. An outer loop running rightward may contain a conditional branch pointing upward, whose true arm contains an inner loop running further right, whose body calls a function in a separate column—each level occupying a distinct spatial axis or region. Because the domain need not be connected, a function body can be placed in a completely separate cluster of cells from its caller; the \instr{K} instruction bridges the gap at run time. Moreover, sequences of \instr{B} (blank) and \instr{T} (turn) instructions can be inserted as needed to avoid crowded regions. 

\begin{proposition}[Arbitrary nesting depth]
\label{prop:nesting}
For every $n \geq 0$, there exists a Grid Program $\GP_n$ whose execution reaches address-stack depth exactly $n$.
\end{proposition}

\begin{proof}
By induction on $n$.

\textbf{Base case ($n = 0$).} The program $\GP_0$ consists of a single cell $(0,0)$ labelled \instr{H}. AS never grows; maximum depth is $0$.

\textbf{Inductive step.} Given $\GP_{n-1}$ with maximum AS depth $n-1$, construct $\GP_n$ as follows. Place a \instr{W} cell at $(0,0)$ with a non-zero initial condition on DS. \instr{W} pushes $(0, 0, \mathsf{up})$ onto AS (depth $1$) and turns the IP clockwise (rightward) into the body. The body consists of the cells of $\GP_{n-1}$, shifted so that their entry point lies at $(1,0)$, with coordinates remapped appropriately. During execution of this embedded sub-program, AS reaches depth $(n-1)+1 = n$, since the entry pushed by \instr{W} remains on the stack throughout the body. After the body the \instr{E} instruction pops the \instr{W} entry; thereafter \instr{W} sees condition $0$ and exits, returning to depth $0$.

The resulting program $\GP_n$ reaches AS depth exactly $n$.
\end{proof}

\begin{remark}
The same argument applies to mixed nesting: a \instr{K} call (depth $+1$) to a subroutine containing a \instr{W} loop (depth $+1$) inside a conditional branch \instr{F} (depth $+1$) reaches AS depth $3$. The spatial layout ensures the three levels do not interfere: each occupies a different direction or a different connected component of the (possibly disconnected) domain.
\end{remark}

\begin{figure}[ht]
  \centering
  \includegraphics[width=15cm]{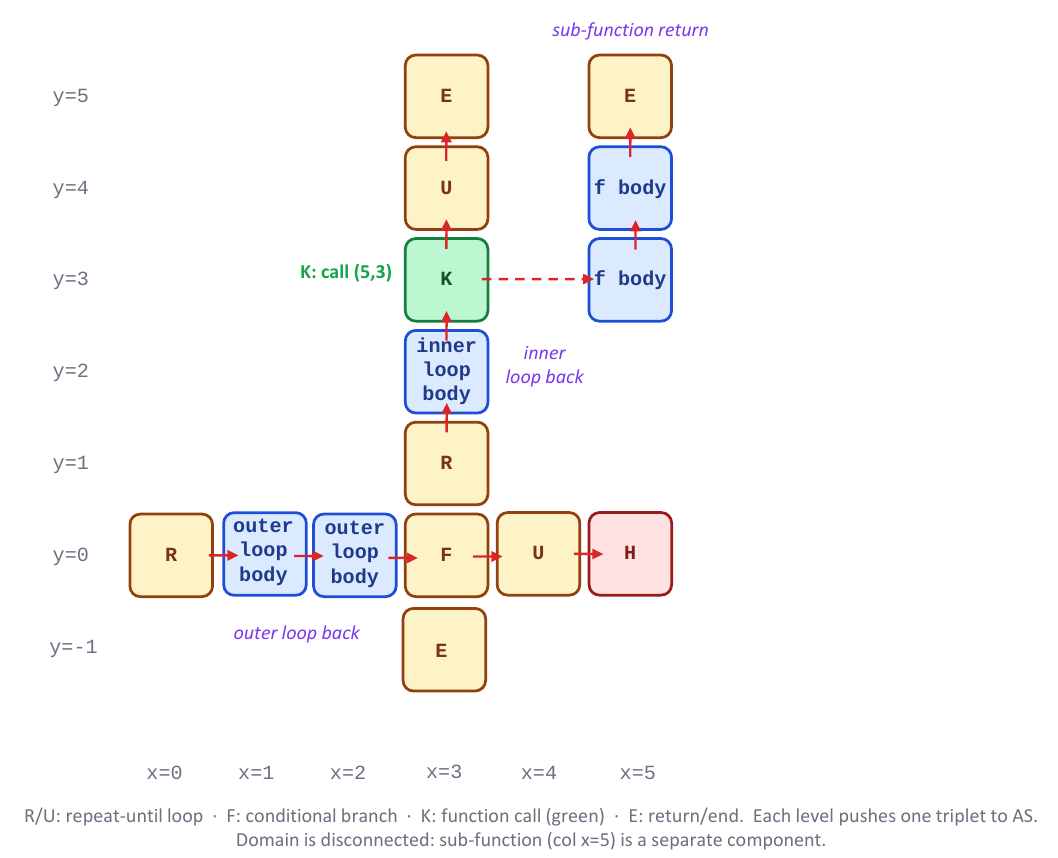}
  \caption{Schematic of four nested control structures: an outer \instr{R}/\instr{U} loop (horizontal), a conditional \instr{F}/\instr{E} branch (upward), an inner \instr{W} loop (upward), and a \instr{K}/\instr{E} function call (rightward) to a subroutine in a separate column. The AS depth annotations on the right show the stack growing by one for each level entered. The domain of this program is disconnected: the subroutine cells (column $x=5$--$6$) form an isolated component joined to the main body only via the run-time \instr{K} jump.}
  \label{fig:nesting_schematic}
\end{figure}

\section{Turing Completeness}
\label{sec:turing}

\begin{theorem}[Turing completeness of Grid Programs]
\label{thm:tc}
Grid Programs are Turing-complete: for every partial computable function $\varphi: \ZZ^* \rightharpoonup \ZZ^*$, there exists a Grid Program $\GP$ that computes $\varphi$.
\end{theorem}

\begin{proof}
We prove Turing completeness by simulating an arbitrary 2-counter machine (Minsky machine) \citep{minsky1967computation}. A 2-counter machine $M$ has two non-negative integer counters $c_1, c_2$ and a finite set of instructions of two kinds:
\begin{enumerate}
  \item $\mathtt{INC}(c_i, \ell)$: increment counter $c_i$; go to instruction $\ell$.
  \item $\mathtt{DEC}(c_i, \ell_1, \ell_2)$: if $c_i > 0$, decrement $c_i$ and go to $\ell_1$; else go to $\ell_2$.
\end{enumerate}
Since 2-counter machines are Turing-complete \citep{minsky1967computation}, it suffices to simulate $M$ with a Grid Program.

\paragraph{Encoding the counters.}
We store $c_1$ and $c_2$ in two CDLL nodes pointed to by $\Prim$ and $\Sec$, respectively. Both are initialised to the input values. The index of the next gadget to be executed is stored in a CDLL node pointed to by $\Ter$.

\paragraph{Encoding program counter.}
Let $M$ have instructions $\ell_0, \ell_1, \ldots, \ell_{p-1}$. We lay out $p$ \emph{gadgets}—one per instruction—along the positive $y$-axis at $x=0$. Each gadget occupies a horizontal strip of cells. The IP starts at $(0,0)$ and passes through gadget $\ell_j$ when it is at row $y = j \cdot h$ for a fixed gadget height $h$.

\paragraph{Simulating \texttt{INC}$(c_1, \ell)$.}
The gadget at row $y_j$ for this instruction:
\begin{enumerate}
  \item \instr{Lprim}: push $c_1$ onto DS.
  \item \instr{P1}: push $1$.
  \item \instr{A+}: DS $\leftarrow c_1 + 1$.
  \item \instr{Sprim}: store incremented value back to $\Prim$.
  \item Set the CDLL node pointed by $\Ter$ to the index of gadget $\ell$.
\end{enumerate}

Simulating \texttt{INC}$(c_2, \ell)$ is analogous to the above, changing \instr{Lprim} to \instr{Lsec} and \instr{Sprim} to \instr{Ssec}.

\paragraph{Simulating \texttt{DEC}$(c_1, \ell_1, \ell_2)$.}
\begin{enumerate}
  \item \instr{Lprim}: push $c_1$.
  \item \instr{P0}: push $0$.
  \item \instr{A>}: push $c_1 > 0$.
  \item \instr{F}: branch on condition.
  \begin{itemize}
    \item \textit{True branch} ($c_1 > 0$): \instr{X} (pop condition); \instr{Lprim}; \instr{P1}; \instr{Asub}; \instr{Sprim}; set the CDLL node pointed by $\Ter$ to the index of gadget $\ell_1$.
    \item \textit{False branch} ($c_1 = 0$): \instr{X}; set the CDLL node pointed by $\Ter$ to the index of gadget $\ell_2$.
  \end{itemize}
\end{enumerate}

Simulating \texttt{DEC}$(c_2, \ell_1, \ell_2)$ is analogous to the above, changing \instr{Lprim} to \instr{Lsec} and \instr{Sprim} to \instr{Ssec}.

\paragraph{Navigation between gadgets.}
Each gadget is structured as a function. The main body of the program is an infinite while loop (\instr{W}/\instr{E}) whose main body uses \instr{K} to call the function (gadget) indicated by the CDLL node pointed by $\Ter$. The execution of each gadget finishes with the return instruction \instr{E}. Since the number of instructions $p$ is finite and fixed, the resulting grid is finite.

\paragraph{Conclusion.}
The constructed grid program correctly simulates $M$: each step of $M$ corresponds to a bounded number of steps of the grid program. Halt in $M$ corresponds to \instr{H} in the grid program. Since $M$ was an arbitrary 2-counter machine, and every partial computable function is computable by some 2-counter machine, the Grid Program model is Turing-complete.
\end{proof}

\begin{corollary}
The halting problem for Grid Programs is undecidable.
\end{corollary}

\section{Related Work}
\label{sec:related}

\paragraph{Two-dimensional programming languages.}
The most well-known two-dimensional language is \textbf{Befunge} \citep{befunge}, introduced by Chris Pressey in 1993. Like Grid Programs, Befunge lays instructions on a two-dimensional grid and uses a directional instruction pointer; control-flow instructions change direction, and values are manipulated on a stack. Befunge-93 operates on a fixed $80 \times 25$ toroidal grid, limiting it to finite computation; Funge-98 \citep{funge98} generalises this to arbitrary dimensions and an infinite grid, achieving Turing completeness. \textbf{Grid Programs} differ from Befunge in several key respects: (i) programs are arbitrary finite (not necessarily connected) subsets of $\ZZ^2$ rather than a fixed rectangular grid; (ii) the CDLL provides structured, pointer-based storage not present in Befunge; (iii) there is a formal address stack and explicit loop/branch structures (\instr{W}/\instr{R}/\instr{U}/\instr{F}/\instr{E}) rather than relying on self-modification of the grid.

\textbf{Piet} \citep{piet} is another notable two-dimensional language, in which programs are bitmap images and instructions are encoded as colour transitions between regions. While visually striking, Piet's instruction set and control-flow model are quite different from Grid Programs.

\textbf{Langton's Ant} \citep{langton1986studying} and other two-dimensional cellular automata \citep{wolfram2002new} demonstrate that complex computation can arise from simple two-dimensional rules, though they are not programming languages in the conventional sense.

\paragraph{Stack-based languages.}
The data stack and postfix arithmetic of Grid Programs are directly inspired by \textbf{Forth} \citep{brodie1984starting}, a concatenative, stack-based language with a long history in embedded systems. \textbf{PostScript} \citep{adobe1999postscript} and \textbf{Factor} \citep{factor} are related stack-based languages. The key novelty in Grid Programs is the combination of a stack with the CDLL and the two-dimensional execution model.

\paragraph{Linked-list and pointer-based models.}
The CDLL in Grid Programs is reminiscent of the tape in a Turing machine \citep{turing1936computable}, except that it is doubly linked and circular, supports three simultaneous access pointers, and allows insertion and deletion. The use of multiple simultaneous pointers into a list recalls \textbf{Pointer machines} \citep{schonhage1980storage} and \textbf{Kolmogorov--Uspensky machines} \citep{kolmogorov1958definition}, which compute over linked structures rather than sequential tapes.

\paragraph{Register machines and counter machines.}
The proof of Turing completeness in Section~\ref{sec:turing} uses 2-counter machines \citep{minsky1967computation}, which are a standard minimal model for establishing Turing completeness. \textbf{RAM models} \citep{cook1972time} provide a closer analogy to practical computation but require explicit addressing; Grid Programs deliberately avoid explicit addresses.

\paragraph{Visual and spatial programming.}
Several visual programming environments, such as \textbf{LabVIEW} \citep{labview} and \textbf{Scratch} \citep{scratch}, lay out programs as two-dimensional diagrams, but these are typically dataflow graphs rather than instruction grids. \textbf{Dataflow architectures} \citep{dennis1974first} similarly use a spatial metaphor for computation but operate at the hardware level. Grid Programs combine spatial layout with an imperative, instruction-pointer-based execution model that is distinct from pure dataflow.

\paragraph{Esoteric and unconventional languages.}
The esolang community \citep{esolangs} has produced numerous languages with unusual control-flow and spatial properties, including \textbf{Malbolge} \citep{malbolge}, \textbf{Whitespace} \citep{whitespace}, and many others. Grid Programs belong to this broader tradition of exploring the design space of computation, but are distinguished by their formal treatment and proof of Turing completeness.

\section{Discussion}
\label{sec:discussion}

\subsection{Expressiveness and Programming Style}

Grid Programs encourage a distinctly spatial programming style. Loops become visible as U-shaped or rectangular detours in the grid; conditionals appear as T-intersections where the IP can branch left or right. Because there are no variable names, all data manipulation must be done through stack operations and CDLL pointer moves. This enforces a discipline similar to concatenative programming \citep{von2011concatenative}, in which the meaning of a program fragment is its effect on a shared implicit state.

The absence of syntax rules is both liberating and challenging. Any grid is a valid program, which means there is no syntax error to catch at parse time—only semantic misbehavior at run time. This mirrors assembly language and bytecode formats, which similarly have no syntactic grammar beyond instruction encoding. For programs intended to be human-authored, the designer must impose their own conventions.

\subsection{Design Variants and Extensions}

Several natural extensions of the base model merit investigation:

\begin{enumerate}
  \item \textbf{Multiple instruction pointers.} Allowing several IPs to move simultaneously, potentially interacting through shared CDLL state, would support a natural model of concurrency.
  \item \textbf{Three-dimensional grids.} Extending $\Dom \subseteq \ZZ^3$ and allowing six directions of travel would further increase the expressiveness of spatial control flow.
  \item \textbf{Self-modifying grids.} If the program can write new instructions into cells of $\Dom$ during execution (cf.\ Befunge's \texttt{p}/\texttt{g} instructions), the model gains reflective capabilities.
  \item \textbf{Typed values.} A type system for the data stack and CDLL would enable static analysis of Grid Programs despite the lack of variable names.
  \item \textbf{Sub-programs (gadgets).} A subroutine mechanism allowing one grid program to call another—analogous to procedure calls—would support modular programming at the grid level.
\end{enumerate}

\subsection{Open Problems}

\begin{enumerate}
  \item \textbf{Instruction set optimization.} The considered instruction set $\Inst$ may be optimised, both by inserting or removing some instructions, or changing their semantics. The goal would be to obtain shorter, easier-to-understand programs.
  \item \textbf{Minimal instruction sets.} What is the smallest subset of $\Inst$ that preserves Turing completeness? In particular, can \instr{K} be removed without loss (since function calls can be simulated by inlining), or does it provide an exponential succinctness advantage?
  \item \textbf{Call depth vs.\ grid size.} What is the maximum call/nesting depth achievable by a Grid Program whose domain fits in an $n \times n$ bounding box?
  \item \textbf{Spatial complexity.} Can every computable function be computed by a Grid Program whose domain fits within a strip of width $w$ for some fixed $w$?
  \item \textbf{Compilation.} Can Grid Programs be efficiently compiled to conventional instruction-set architectures, and what is the overhead of simulating the CDLL and spatial control flow?
\end{enumerate}

\section{Conclusion}
\label{sec:conclusion}

We have introduced Grid Programs, a new model of computation characterised by a planar program structure, the absence of syntax constraints, and the absence of named variables. The instruction pointer navigates a finite two-dimensional grid; control flow is encoded as directional changes in the IP's path; data is managed through a stack and a circularly doubly linked list. We have given a formal operational semantics, worked through five detailed examples (absolute value, factorial, while-loop sum, linear search, string reversal), proved Turing completeness via simulation of 2-counter machines, surveyed related work, and identified open problems.

Potential applications of this model include visual teaching of computer programming, hardware implementations that exploit the two-dimensional connectivity of the programs, and source code analysis and generation by deep learning systems. The grid nature of the programs might be amenable to their processing by deep learning models that specialize on two dimensional data.

Grid Programs occupy a unique position in the landscape of computational models: they are as expressive as any Turing-complete formalism, yet their structure is radically different from that of linear, variable-based programming. We hope that this formalism will stimulate further research into two-dimensional computation, spatial programming metaphors, and the fundamental question of what a program can look like.

\bibliographystyle{unsrtnat}
\bibliography{references}

\end{document}